\documentclass[prx,twocolumn,superscriptaddress,longbibliography]{revtex4-2}
\pdfoutput=1

\usepackage{amsmath,braket,amssymb}
\usepackage{amsthm}
\usepackage[final]{graphicx}
\usepackage{subfigure}
\usepackage{xcolor}
\usepackage[english]{babel}
\usepackage{color}
\usepackage{mathtools}
\usepackage{physics}
\usepackage{empheq}
\usepackage{tensor}

\usepackage{hhline}

\usepackage[colorlinks,citecolor=darkBlue,linkcolor=darkBlue,
urlcolor=blue,hyperindex]{hyperref}

\usepackage[normalem]{ulem}
\normalfont

\usepackage[linesnumbered, ruled,vlined]{algorithm2e}

\graphicspath{{./images/},{./imagesAppendix/}}

\definecolor{forestgreen}{rgb}{0.08, 0.4, 0.13}
\definecolor{darkBlue}{rgb}{0.08, 0.13, 0.4}

\definecolor{THc}{rgb}{0.7,0.2,0.7}

\newcommand\numberthis{\addtocounter{equation}{1}\tag{\theequation}}

\newtheorem{theorem}{Theorem}
\newtheorem{corollary}{Corollary}
\newtheorem{definition}{Definition}
\newtheorem{conjecture}{Conjecture}
\newtheorem{fact}{Fact}
\newtheorem{lemma}{Lemma}

\newcommand{\SM}{Appendix}

\renewcommand{\eqref}[1]{Eq.~(\ref{#1})} 

\newcommand{\subfigimg}[3][,]{%
	\setbox1=\hbox{\includegraphics[#1]{#3}}%
	\leavevmode\rlap{\usebox1}%
	\rlap{\hspace*{2pt}\raisebox{\dimexpr\ht1-0.5\baselineskip}{{\bfseries \large\textsf{#2}}}}%
	\phantom{\usebox1}%
}

\begin{document}

\title{Triangle Criterion: a mixed-state magic criterion with applications in distillation and detection}

\author{Zhenhuan Liu}
\email{qubithuan@gmail.com}
\thanks{equal contribution}
\noaffiliation

\author{Tobias Haug}
\email{tobias.haug@u.nus.edu}
\thanks{equal contribution}
\affiliation{Quantum Research Center, Technology Innovation Institute, Abu Dhabi, UAE}

\author{Qi Ye}
\email{yeq22@mails.tsinghua.edu.cn}
\affiliation{Center for Quantum Information, IIIS, Tsinghua University, Beijing, China}
\affiliation{Shanghai Qi Zhi Institute, Shanghai, China}

\author{Zi-Wen Liu}
\email{zwliu0@tsinghua.edu.cn}
\affiliation{Yau Mathematical Sciences Center, Tsinghua University, Beijing 100084, China}

\author{Ingo Roth}
\email{ingo.roth@tii.ae}
\affiliation{Quantum Research Center, Technology Innovation Institute, Abu Dhabi, UAE}
	
\date{\today}
	
\begin{abstract}
We introduce a mixed-state magic criterion, the \emph{Triangle Criterion}, which plays a role for magic analogous to the Positive Partial Transposition (PPT) Criterion for entanglement: it combines strong detection capability, a clear geometric interpretation, and an operational link to magic distillation. 
Using this criterion, we uncover several new features of multi-qubit magic distillation and detection. 
We prove that genuinely multi-qubit magic distillation protocols are strictly more powerful than all single-qubit schemes by showing that the Triangle Criterion is not stable under tensor products.
Moreover, we show that, with overwhelming probability, multi-qubit magic states with relatively low rank cannot be distilled by any single-qubit distillation protocol.
We derive an upper bound on the minimal purity of magic states, which is conjectured to be tight with both numerical and constructive evidences. 
Using this minimal-purity result, we predict 
the existence of unfaithful magic states, namely states that cannot be detected by any fidelity-based magic witness, and reveal  fundamental limitations of mixed-state magic detection in any single-copy scheme.
\end{abstract}
	
\maketitle

\section{Introduction}
Entanglement and magic (or non-stabilizerness) are central non-classical resources that empower quantum communication and computation beyond what is possible classically~\cite{Horodecki2009EntanglementRMP,Veitch2014ResourceTheoryStabilizer}.
Entanglement has been studied for decades, leading to a well-developed toolbox of entanglement measures and witnesses that clarify both the geometry of entangled states and their operational role in information-processing tasks~\cite{PlenioVirmani2007EntanglementMeasures,GuhneToth2009EntanglementDetection}. 
By contrast, the theory of magic is comparatively young: after the Gottesman--Knill theorem~\cite{gottesman1998heisenberg}, it became clear that non-stabilizer ``magic'' states are the computational resources that turn classically simulable stabilizer circuits into universal quantum computers~\cite{bravyi2005universal,Veitch2012NegativeQuasiProbability,HowardCampbell2017RoM}. 
Despite the surging recent interest and efforts on magic theory, most progress has focused on the quantification aspect within the framework of resource theories~\cite{HowardCampbell2017RoM,SeddonCampbell2019QuantifyingMagic,wang2020efficient,leone2022stabilizer,liu2022many}. 
Simple and physically meaningful tools for characterizing magic, especially for general mixed states, are still largely lacking.

A promising route is to draw inspiration from the mature tools developed for entanglement detection, among which the Positive Partial Transpose (PPT) criterion~\cite{peres1996separability} has emerged as a prominent one.
In its simplest form, it asks whether the partial transpose $\rho^{\mathrm{T}_A}$ of a bipartite state $\rho$ is positive semidefinite, where $\mathrm{T}_A$ denotes the partial transposition on subsystem $A$.
Despite this concise mathematical formulation, the PPT criterion enjoys several key properties:
(1) it provides a necessary and sufficient condition for separability of mixed states in $2\times 2$ and $2\times 3$ dimensions and for pure states in arbitrary dimensions~\cite{Horodecki1996Separability};
(2) it is fundamentally linked to bound entanglement: any PPT state is non-distillable and hence cannot be used to prepare Bell states~\cite{Horodecki1998BoundEntanglement,VidalWerner2002Computable};
(3) it yields a tight purity threshold for entanglement, $\mathrm{Tr}(\rho^2)>1/(d-1)$ with $d$ being the dimension of $\rho$, so that all states below this threshold are guaranteed to be separable~\cite{Gurvits2002largest}.
Because of these favourable properties, the PPT criterion has been widely adopted in diverse tasks, including the resource theory of entanglement~\cite{Audenaert2003PPTcost,WangWilde2023ExactPPT}, quantum device benchmarking~\cite{Shaw2024BenchmarkingEntangled60Atom}, and quantum field theory~\cite{Calabrese2012NegativityQFT}.
It is thus natural and important to ask: \emph{is it possible to construct a single, mixed-state magic criterion that reproduces the key advantages of the PPT criterion?}

In this work, inspired by various previous results on magic distillation and stabilizer polytopes~\cite{reichardt2004improved,reichardt2006quantum,okay2021extremal}, we establish the \emph{Triangle Criterion}, a mixed-state magic criterion taking the form of fidelity triangle inequalities, with following properties closely mirroring those of the PPT entanglement criterion:
(1) it provides a necessary and sufficient condition for detecting magic in single-qubit mixed states and in pure states of an arbitrary number of qubits;
(2) a multi-qubit mixed state can be converted into a single-qubit magic state by stabilizer operations (including post-selection) if and only if it is detected by the criterion;
(3) it yields an upper bound on the minimal purity of magic states, $1/(d-1/2)$, which is conjectured to be tight with numerical and constructive evidences.
Despite these similarities, a key difference from PPT is that the Triangle Criterion is not tensor-stable: the tensor product of two states that each satisfy the Triangle Criterion can violate it when viewed as a larger system, leading to nontrivial consequences for magic distillation.

With this magic criterion, we identify a family of two-qubit states that are provably useless for any single-qubit magic distillation protocol: under arbitrary stabilizer processing, every single-qubit state obtainable from one copy is a stabilizer state. 
Nevertheless, we construct a distillation routine that converts multiple copies of these two-qubit states into $T$ states, providing (to our knowledge) the first proof that genuinely multi-qubit magic distillation protocols are strictly more powerful than all single-qubit schemes.
We further show that, for random mixed states with rank polynomial with qubit number, single-qubit magic distillation protocols fail with high probability.
This indicates that single-qubit distillability of multi-qubit states is fragile to noise.
We further establish the existence of \emph{unfaithful magic states}, which cannot be detected by any fidelity-based magic witness. 
In addition, our conjectured lower bound on the minimal purity of magic states implies both the existence of \emph{absolute stabilizer states}, states that cannot be transformed into magic states by any unitary evolution and measurement post-selection, and a \emph{fundamental limitation} on certifying the magic of mixed states.

\section{Triangle Criterion}

An $n$-qubit Pauli operator is an element of the $n$-qubit Pauli group, namely a tensor product of single-qubit Pauli operators $I,X,Y,Z$ up to an overall phase. 
An $n$-qubit stabilizer pure state is defined as the unique simultaneous $+1$ eigenstate of a set of $n$ independent commuting Pauli operators $P_1,\dots,P_n$. 
Equivalently, it is the unique state stabilized by an abelian subgroup of the $n$-qubit Pauli group generated by these operators. 
A general stabilizer state is any convex mixture of stabilizer pure states. The set of all such states forms the stabilizer polytope, namely the convex hull of all stabilizer pure states. In other words, the stabilizer polytope consists precisely of all density matrices that can be written as probabilistic mixtures of stabilizer pure states. In the simplest case of a single qubit, the six pure stabilizer states $\ket{0}, \ket{1}, \ket{+}, \ket{-}, \ket{+i}, \ket{-i}$ form the vertices of an octahedron inside the Bloch sphere. This octahedron is precisely the single-qubit stabilizer polytope.
A Clifford unitary is a unitary operator that maps each Pauli group element to another under conjugation. Equivalently, Clifford unitaries preserve the set of stabilizer states. 
States outside the stabilizer polytope are called magic states. 
That is, magic states are the states that cannot be expressed as a convex mixture of stabilizer pure states.

Inspired of the geometry of single-qubit stabilizer polytope~\cite{garcia2017geometry}, we formulate the following general magic criterion:
\begin{theorem}[Triangle Criterion]\label{thm:magic_PPT}
Given any $n$-qubit quantum state $\rho\in\mathcal{H}_{2^n}$, if there exist three nearest-neighbour stabilizer pure states $\{\psi_1,\psi_2,\psi_3\}$ with $\psi=\ketbra{\psi}{\psi}$ and $\Tr(\psi_i\psi_j)=1/2$ for all $i\neq j$ such that 
\begin{equation}
\Tr(\rho\psi_1)>\Tr(\rho\psi_2)+\Tr(\rho\psi_3)\,,
\end{equation}
then $\rho$ is not a mixture of stabilizer states.
This criterion detects all single-qubit mixed magic states and all multi-qubit pure magic states.
\end{theorem}
\noindent 
Similar in spirit to the PPT criterion, the validity of the Triangle Criterion can be understood as follows, from two simple structural observations about stabilizer states: (i) Given two stabilizer pure states $\psi_1$ and $\psi_2$ with $\Tr(\psi_1 \psi_2)=1/2$, for another stabilizer pure state $\phi$, the ratio of $\Tr(\psi_1\phi)/\Tr(\psi_2\phi)$ can only take the value of $1/2$, $1$, and $2$ when both $\Tr(\psi_1\phi)$ and $\Tr(\psi_2\phi)$ are non-zero;
(ii) Given three stabilizer pure states $\{\psi_1,\psi_2,\psi_3\}$ with $\Tr(\psi_i\psi_j)=\frac{1}{2}$ for all $i\neq j$ and another stabilizer pure state $\phi$, when $\Tr(\psi_1\phi)=0$, we have $\Tr(\psi_2\phi)=\Tr(\psi_3\phi)$.
Therefore, when $\rho$ is a stabilizer pure state, it holds that $\Tr(\rho\psi_1)\le\Tr(\rho\psi_2)+\Tr(\rho\psi_3)$ where the equality is achieved only when $\Tr(\rho\psi_1)=2\Tr(\rho\psi_2)=2\Tr(\rho\psi_3)$, or $\Tr(\rho\psi_2)=0$, or $\Tr(\rho\psi_3)=0$.
This inequality naturally generalizes to the case when $\rho$ is a mixed stabilizer state, i.e., a convex combination of stabilizer pure states.
In the single-qubit case, the equalities $\Tr(\rho\psi_1)=\Tr(\rho\psi_2)+\Tr(\rho\psi_3)$ describe all facets of the stabilizer octahedron and thus give the necessary and sufficient criterion for all single-qubit mixed states.
The detailed proof is given in Appendix~\ref{app:triangle}.

In what follows, we will establish many close parallels between the PPT criterion and the Triangle Criterion. However, it is important to emphasize a key aspect in which the Triangle Criterion is less convenient than the PPT criterion: it is in general computationally demanding. For PPT, one only needs to apply partial transposition and check positivity, which renders the criterion directly computable through spectral decomposition or semi-definite programming.
Furthermore, PPT criterion can also be estimated experimentally through suitable polynomial relaxations~\cite{gray2018machine,elben2020mix,Wang2022DetectingQuantifyingEntanglement,zhang2025entanglement,tarabunga2025quantifying}. By contrast, in its present form, the Triangle Criterion requires searching over $4(2^n-1)2^n\prod_{l=1}^n(2^l+1)\approx 2^{n^2}$ low-rank observables of the form
\begin{equation}
    W_{ijk}=\psi_i+\psi_j-\psi_k\,,
\end{equation}
and checking whether any of them has a negative expectation value.

We view this difference as reflecting an intrinsic structural distinction between entanglement and magic, rather than a weakness of the criterion itself. In entanglement theory, the PPT test arises from a fixed linear transformation that is naturally compatible with the bipartite tensor-product structure. For magic, by contrast, the relevant free set is the stabilizer polytope, whose boundary is determined by many discrete facets, so one should not expect a direct analogue of the PPT procedure with the same computational simplicity and detection capability. This obstruction already appears in the single-qubit case, where unit-trace positive semidefinite operators occupy the region inside the Bloch sphere. If one wishes to detect all single-qubit magic states with a single trace-preserving linear map, then this map must send all quantum states inside the stabilizer octahedron to points inside the Bloch sphere, while sending all quantum states inside the Bloch sphere but outside the octahedron to points outside the Bloch sphere. Such a map does not exist, since no linear map can transform an octahedron into a sphere, i.e. map flat faces to a curved boundary.
Nevertheless, although the Triangle Criterion is less directly computable than PPT, it remains conceptually natural and is strong enough to recover the full single-qubit mixed-state case, all pure-state magic, and the exact characterization of single-qubit magic distillability established in Theorem~\ref{thm:bound_magic}.

\section{Multi-qubit magic distillation}
The most crucial property of the Triangle Criterion is its connection with magic distillation, which starts from the following result:
\begin{theorem}[Magic Distillation]\label{thm:bound_magic}
Given an $n$-qubit state $\rho$, it can be non-deterministically transformed into a single-qubit magic state with Clifford unitary rotations, Pauli measurements with post-selection, and stabilizer ancillas if and only if it can be detected by Triangle Criterion.
\end{theorem}
This theorem can be proved using two key ingredients.
First, we show in \SM{}~\ref{app:equivalence} that the Triangle Criterion is equivalent to estimating the magic witness
\begin{equation}\label{eq:magic_witness}
W_U=U\left[(\ketbra{+}{+}+\ketbra{+i}{+i}-\ketbra{0}{0})\otimes\ketbra{0^{n-1}}{0^{n-1}}\right]U^\dagger,
\end{equation}
where $U$ is chosen from the entire Clifford group.
Notice that $\ketbra{+}{+}+\ketbra{+i}{+i}-\ketbra{0}{0}$ is itself a magic witness and, after suitable single-qubit Clifford rotations, can be used to detect all single-qubit magic states.
Estimating $\Tr(W_U\rho)$ is therefore operationally equivalent to first rotating $\rho$ with a Clifford unitary, then projecting $n-1$ qubits onto the state $\ket{0^{n-1}}$, and finally testing the magic of the resulting single-qubit state.
Consequently, if a state is detected by the Triangle Criterion, there exists a stabilizer operation that reduces it to a single-qubit magic state after the projection.
Since Clifford unitaries and projections onto $\ket{0^{n-1}}$ cannot create magic, this also establishes the correctness of the Triangle Criterion in the multi-qubit setting.

The second ingredient comes from Lemma~1 of Ref.~\cite{reichardt2006quantum}, which shows that if a stabilizer procedure (with measurement post-selection and stabilizer ancillas) can map an $n$-qubit state $\rho$ into a single-qubit magic state, then there exists another procedure that maps $\rho$ to a single-qubit magic state and consists solely of an $n$-qubit Clifford unitary followed by projecting $n-1$ qubits onto $\ket{0^{n-1}}$, which is equivalent to estimating \eqref{eq:magic_witness}.
We can therefore conclude that, if there exists a stabilizer operation that maps $\rho$ to a single-qubit magic state, then $\rho$ can be detected by the Triangle Criterion, which completes the proof of Theorem~\ref{thm:bound_magic}.
Moreover, Theorem~\ref{thm:bound_magic} also implies that the Triangle Criterion detects all pure magic states, since any multi-qubit pure magic state can be reduced to a single-qubit pure magic state by a suitable stabilizer procedure~\cite{reichardt2004improved}.
This means that, for pure-state magic detection, the Triangle Criterion is equivalent to other criteria such as stabilizer fidelity~\cite{Bravyi2019simulationofquantum}, stabilizer nullity~\cite{Beverland2020nullity}, and stabilizer Rényi entropy~\cite{leone2022stabilizer}, in the sense that they have exactly the same detection power on pure states.

Note that Theorem~\ref{thm:bound_magic} is similar with the relationship between PPT and entanglement distillation: a PPT entangled state cannot be used for preparing Bell states.
In entanglement theory, the PPT criterion is used not only for detecting state entanglement, but also for characterizing the entanglement-generation power of a quantum channel, based on the key observation that the Choi state of an entanglement-breaking channel is separable~\cite{horodecki2003entanglementbreaking}.
Therefore, a natural question is whether the Triangle Criterion can also characterize the properties of quantum channels.
Note that preparing the Choi state of a given channel using the query access to the channel does not require additional magic. This leads to the following result, whose detailed proof is given in Appendix~\ref{app:channel}:
\begin{corollary}\label{coro:channel}
Given a quantum channel $\mathcal{C}:\mathcal{H}_{2^n}\to\mathcal{H}_{2^m}$, it can be used to non-deterministically produce a single-qubit magic state by stabilizer operations, including stabilizer-state inputs, Clifford unitary rotations, and Pauli measurements with post-selection, if and only if its Choi state $\mathcal{I}_n\otimes\mathcal{C}(\Phi^+_n)$ can be detected by the $(n+m)$-qubit Triangle Criterion, where $\mathcal{I}_n$ is the $n$-qubit identity channel and $\ket{\Phi^+_n}=\frac{1}{\sqrt{2^n}}\sum_{i=0}^{2^n-1}\ket{ii}$ is the maximally entangled state.
\end{corollary}

Despite the similarity, a major difference between PPT and Triangle criteria is the stability under tensor product.
It is easy to show that if two states are both PPT, then the tensor product of these two states is also PPT.
This means that there exists some bound entangled state which cannot be transformed into Bell states given arbitrary number of copies.
However, this conclusion cannot be directly generalized to Triangle Criterion as the definition of magic has no space structure, i.e., it is not compatible with the tensor product operation.
Actually, we prove that the Triangle Criterion is not tensor stable:
\begin{theorem}\label{thm:activation}
There exists a two-qubit state which cannot be converted, even probabilistically, into single-qubit magic states with stabilizer operations using a single copy.
However, the state can be used to distil arbitrary single-qubit magic states given multiple copies.
\end{theorem}
This conclusion indicates that genuinely multi-qubit magic distillation protocols are strictly more powerful than single-qubit distillation protocols, thereby resolving the open problem posed in Ref.~\cite{reichardt2006quantum}. 
Here, a single-qubit (multi-qubit) magic distillation protocol takes many single-qubit (multi-qubit) magic states as input and outputs a state close to some pure single-qubit magic state, such as the $T$ state defined as $\ket{T}\propto\ket{0}+e^{\mathrm{i}\pi/4}\ket{1}$.
We establish this result by numerically searching over all two-qubit states that are not detected by the two-qubit Triangle Criterion but are detected by the four-qubit Triangle Criterion with two copies.
By optimizing the violation of the four-qubit Triangle Criterion, we find a two-qubit state such that two copies of this state can be converted into a single-qubit magic state, which can then be further distilled into a $T$ state~\cite{bravyi2005universal}.
The explicit state we find and the corresponding distillation procedure are presented in \SM{}~\ref{app:single_distillation}.

Having seen that single-qubit magic distillation protocols can fail for some multi-qubit states, a natural question is how often single-qubit protocols are actually useful.
For a given multi-qubit state, a single-qubit distillation protocol can be useful only if the state can be reduced to a single-qubit magic state by stabilizer operations (the converse would rely on resolving the open problem of the existence of bound magic under arbitrary numbers of copies~\cite{campbell2010bound}).
Since Theorem~\ref{thm:bound_magic} shows that the Triangle Criterion is a necessary and sufficient condition for a multi-qubit state to be reducible to a single-qubit magic state, we can sample states according to some mixed-state distribution and compute the probability that they are detected by the Triangle Criterion.

Hereafter, we will adopt the following mixed-state distribution to analyze the probabilities of multi-qubit magic detection and distillation~\cite{zyczkowski2001induced}, which is defined naturally by taking the partial trace of the uniform distribution of pure states:
\begin{definition}
A $d$-dimensional mixed state $\rho$ is sampled from the distribution $\pi_{d,k}$ when $\rho=\Tr_k(\ketbra{\Psi}{\Psi})$ where $\ket{\Psi}$ a $d\times k$-dimensional pure state sampled from the Haar measure.
\end{definition}
\noindent Leveraging tools developed in Ref.~\cite{liu2022detectability}, we prove:
\begin{theorem}\label{thm:single_utility}
Given a $d=2^n$-dimensional state $\rho$ sampled according to the distribution of $\pi_{d,k}$, the probability that it can be distilled into single-qubit $T$\,state by some single-qubit magic distillation protocol is upper bounded by $\exp{c_1n^2-c_2k}$, where $c_1$ and $c_2$ are constants.
\end{theorem}
\noindent Since $k$ upper-bounds the rank of $\rho$, it implies that when the rank of a random mixed state reaches the order of $n^2$, single-qubit distillation fails with overwhelming probability. In stark contrast, magic can persist even for $k$ on the order of $2^n$~\cite{bansal2025pseudo}, exhibiting an exponential gap between the existence of magic and its distillability by single-qubit protocols.

\section{Minimal purity}
In entanglement theory, an important conclusion is that the minimal purity of entangled $d$-dimensional state is arbitrarily close to $\frac{1}{d-1}$, which plays vital roles in entanglement detection~\cite{Barreiro2010multipartite} and geometric analysis~\cite{aubrun2012phase}.
The upper bound is given by the PPT criterion while it has also been shown to be the lower bound in Ref.~\cite{Gurvits2002largest}.
Despite the distinct structures associated with entanglement and magic theories, we use the Triangle Criterion to establish a remarkably similar characteristic:
\begin{theorem}\label{thm:minimal_purity}
There exists a multi-qubit magic state with purity arbitrarily close to $\frac{1}{d-1/2}$.
\end{theorem}
This theorem is proved by exhibiting a state on the surface of the polytope defined by the Triangle Criterion, i.e., with $\Tr(W_{ijk}\rho)=0$, that has minimal purity.
A representative state is
\begin{align*}
\rho=&\frac{1}{d-1/2}I^{\otimes n}+\frac{1/2}{d-1/2}(\ketbra{0}{0}\\
&+\ketbra{+}{+}+\ketbra{+i}{+i}-2I)\otimes\ketbra{0^{n-1}}{0^{n-1}}\,.\numberthis
\end{align*}
We prove in \SM{}~\ref{app:purity} that this state actually lies on the boundary of the stabilizer polytope, in agreement with the conclusion of Ref.~\cite{okay2021extremal}.
Hence, if we move this state towards the maximally mixed state in order to further reduce its purity, it will eventually enter the stabilizer polytope and become a stabilizer state.
At the same time, one can check that this density matrix is strictly positive.
Therefore, it is not on the boundary of the state space, and there exists a magic state outside the stabilizer polytope that is arbitrarily close to it.

We conjecture that: 
\begin{conjecture}\label{conj:minimal_purity}
Any $d$-dimensional multi-qubit state $\rho$ with purity less equal than $\frac{1}{d-1/2}$ can be written as a mixture of stabilizer pure states.
\end{conjecture}
\noindent In single and two-qubit cases, we can numerically enumerate all different surfaces of stabilizer polytopes to test the correctness of our conjecture.
For two qubits, the stabilizer polytope has eight different kinds of surfaces which cannot be transformed into each other with Clifford unitaries~\cite{reichardt2006quantum,junior2025geometric}.
Among them, one kind of surface is equivalent with the witness of Triangle Criterion and gives the purity of $\frac{1}{d-1/2}=\frac{2}{7}$.
All the other seven surfaces have larger purities.
For even more qubits, we have numerical evidence for the conjecture up to $n=4$.
It is worth mentioning that, although $\frac{1}{d-1/2}$ is close to $\frac{1}{d}$, the constant $1/2$ is crucial: it leads to the fundamental limitation in magic detection shown in Theorem~\ref{thm:linear}, whereas this limitation would no longer hold if the constant were exponentially small in $n$.

An intriguing corollary of this conjecture is the existence of absolute stabilizer states.
We prove in \SM{}~\ref{app:absolute} that states with purity $\frac{1}{d-c}$ with $c$ being some constant are stable under measurement post-selection:
\begin{theorem}\label{thm:purity_reduction}
Given an arbitrary $d=d_A\times d_B$-dimensional bipartite state $\rho_{AB}$ with purity $\frac{1}{d-c}$ with $0\le c\le d_A-1$, after arbitrary unitary performed on $AB$ and arbitrary measurement and post-selection performed on $B$, the maximal purity of the reduced density matrix on subsystem $A$ is $\frac{1}{d_A-c}$. 
\end{theorem}
\noindent This means that, for any resource, including entanglement and magic, if the minimal purity of resourceful state has the form of $\frac{1}{d-c}$, then there exists an \emph{absolutely non-resourceful} ball around the maximally mixed state such that given any state in the ball, any operation including unitary rotation and measurement with post-selection of a subsystem cannot make the reduced state resourceful.
To avoid misunderstanding, we emphasize that Theorem~\ref{thm:purity_reduction}, together with Theorem~\ref{thm:unfaithful_magic} introduced later, does not rely on Conjecture~\ref{conj:minimal_purity}. Theorem~\ref{thm:linear} is the only result that depends on this conjecture.

\section{Multi-qubit magic detection}
Witnesses are a powerful and experimentally attractive tool for resource detection, corresponding to hyperplanes that separate free and resourceful states.
A magic (entanglement) witness is an observable whose expectation values are positive for all stabilizer (separable) states while negative for some magic (entangled) states.
According to the hyperplane separation theorem, whenever the set of free states is convex, witnesses are complete for resource detection, namely, for any resourceful state, there exists a witness that can detect it. This conclusion naturally applies to magic detection~\cite{leone2026unbearable}. However, this is only an existential statement, and the construction of such witnesses might highly dependent on the given state and might require using general observables without any further structure.
In this section, instead of asking whether some witness exists for a given magic state, we study the detection power and limitations of fixed witnesses, witnesses with specific structures, and more general linear detection protocols built from finitely many observables.

A commonly-used class of witnesses is based on fidelity, $W = \alpha I - \ketbra{\psi}{\psi}$, built on the simple idea that a resourceful state should have high fidelity with some resourceful pure state~\cite{pan2012multiphoton}.
For entanglement, $\ket{\psi}$ is typically chosen to be a maximally entangled state and $\alpha$ is determined by the Schmidt coefficients of $\ket{\psi}$.
For magic, $\ket{\psi}$ can be taken as a highly magical pure state and $\alpha$ as its maximal fidelity with a stabilizer pure state.
In entanglement detection, it has been shown that there exist unfaithful entangled states that cannot be detected by any fidelity-type witness~\cite{weilenmann2020faithful,gunhe2021faithful}, revealing that some entangled states are far from all entangled pure states.
We find such a phenomenon also for magic and  construct an unfaithful magic state using Theorem~\ref{thm:minimal_purity} in \SM{}~\ref{app:faithful}.
\begin{theorem}\label{thm:unfaithful_magic}
There exist magic states that cannot be detected by any magic witness of the form $\alpha I-\ketbra{\psi}{\psi}$.
\end{theorem}

Actually, magic detection is not only hard using fidelity-type witnesses.
If Conjecture~\ref{conj:minimal_purity} holds, then magic detection is generally subject to fundamental limitations, which would hold for all linear protocols even beyond magic witnesses.
Here, by a linear protocol we mean that one decides whether the target state $\rho$ is a magic state or not from the expectation values of $M$ different observables, $\{\Tr(O_i\rho)\}_{i=1}^M$.
We prove in \SM{}~\ref{app:detectability} that for mixed states sampled with $\pi_{d,k}$, a number of $\tilde{\Omega}(k)$ linear observables is required to detect magic with constant probability~\cite{liu2022detectability}.
The notation $\tilde{\Omega}(f(n))$  is an asymptotic lower bound that ignores polylogarithmic factors.
Note that $k$ is the dimension of the traced-out system, which can be exponentially large in practical scenarios.
\begin{theorem}\label{thm:linear}
If Conjecture~\ref{conj:minimal_purity} is correct, then there exists a constant $c_1>0$ such that for any magic witness, the probability of successfully detecting magic for a mixed state sampled from $\pi_{d,k}$ is upper bounded by $\exp{-c_1 k}$.
Moreover, there exists a constant $c_2>0$ such that the probability that the magic of the state can be detected from the expectation values of $M$ different observables is upper bounded by $\exp{M\ln(4\sqrt{M}d)-c_2 k}$.
\end{theorem}

\section{Discussion}
In this work, we have introduced a powerful mixed-state magic criterion and used it to reveal several new properties of multi-qubit magic distillation and detection. 
We hope that our work will inspire further exploration of mixed-state magic detection and quantification, as well as multi-qubit magic distillation.
Beyond proposing a new tool for magic detection, our work has a promising conceptual conclusion: although entanglement and magic are two resources with very different mathematical formulations, both admit remarkably simple criteria that capture closely analogous structural properties. This parallel may point to a deeper connection between the two prominent quantum resources.
At the same time, our results raise a number of open questions.
The most immediate one concerns the proof of Conjecture~\ref{conj:minimal_purity}~\cite{zurel2026basis}.
In \SM{}~\ref{app:purity}, we rigorously prove that all states with purity less than $\frac{1}{d-1/d}$ must be stabilizer states, with two methods of Pauli decomposition and mutually unbiased bases~\cite{Bengtsson2005ComplementarityPolytope,Zhu2015MUBClifford2Design}. 
We believe that a tighter lower bound can be obtained by generalizing the Pauli-decomposition argument in \SM{}~\ref{app:purity}, where we only use the observation that $I+P$ is a stabilizer mixed state.
By exploiting more detailed structure of stabilizer states, one may use the identity operator to cancel additional Pauli terms and thereby derive a sharper bound.

Theorem~\ref{thm:single_utility} shows that there is a broad range of ranks, from $n^2$ to $2^n$, in which a multi-qubit state still possesses magic but cannot be transformed into a single-qubit magic state via any single-copy protocol. At the same time, Theorem~\ref{thm:activation} shows that this ability to transform the state into a single-qubit magic state can be activated by using multiple copies. A natural and interesting open question is under what conditions a state can always be transformed into a single-qubit magic state given sufficiently many copies. In particular, is this condition equivalent to the presence of multi-qubit magic, in analogy with the activation of genuine multipartite entanglement~\cite{PalazuelosVicente2022GMEActivation} where any state that is not separable across some fixed bipartition can exhibit genuine multipartite entanglement when sufficiently many copies are provided?

Another promising direction is to develop computable non-linear quantities inspired by the Triangle Criterion to approximate its detection power. As discussed above, a main practical limitation of the Triangle Criterion is that, in its current form, it relies on a discrete family of witnesses and is therefore less directly computable than PPT. A natural way to improve this is to replace the exact criterion by suitably designed non-linear functions of the state, for example magic-sensitive polynomial or entropy-like quantities such as stabilizer R\'enyi entropies~\cite{leone2022stabilizer,haug2025efficient}. Such quantities may provide tractable and experimentally accessible proxies for the Triangle Criterion, and may also achieve detection tasks that no single linear witness can accomplish, for example detecting all pure magic states with a single non-linear function. In this sense, they may offer a natural route to overcoming the limitations of linear detection methods established in Theorem~\ref{thm:linear}. At the same time, because the Triangle Criterion has a different mathematical structure from PPT, carrying out such a reduction will likely require new ideas rather than a direct adaptation of existing polynomial-relaxation methods from entanglement detection~\cite{yu2021optimal}.

A further direction is to enhance the detection capability of the Triangle Criterion itself.
For entanglement, the PPT criterion can be systematically strengthened by the symmetry-extension method~\cite{Doherty2004complete}, yielding a necessary and sufficient separability test at the cost of increased computational complexity.
As shown in Theorem~\ref{thm:bound_magic}, the Triangle Criterion is in fact equivalent to first reducing the state to a single-qubit state and then testing its magic.
This perspective naturally suggests a hierarchy of stronger criteria obtained by reducing the state to effective systems with progressively larger numbers of qubits. Such extensions can improve the detection power of the original Triangle Criterion and, at the same time, have a natural operational connection to magic distillation tasks targeting multi-qubit resource states, such as the CCZ state~\cite{shi2002both,Haah2018codesprotocols}. In addition, compared with the original Triangle Criterion, such multi-qubit extensions would naturally involve a broader class of stabilizer states in the construction of the criterion, including stabilizer pairs that are not nearest neighbours and may have mutual fidelity smaller than $1/2$.
Similar to symmetry extensions for the PPT criterion, this stronger detection power is expected to come at the cost of increased computational complexity.

Inspired by entanglement negativity based on the PPT criterion~\cite{VidalWerner2002Computable} and Wigner negativity~\cite{Kenfack2004WignerNegativity,Gross2006HudsonFiniteDim}, we define a magic measure, the \emph{Triangle Negativity}, based on our criterion:
\begin{equation}
    \mathcal{T}(\rho)=\log\!\left(\sum_{ijk}\frac{\vert\Tr(W_{ijk}\rho)\vert}{4(2^n-1)\prod_{\ell=1}^n(2^\ell+1)}\right).
\end{equation}
This measure fulfils the desired properties that $\mathcal{T}(\rho)>0$ if and only if there exists a stabilizer protocol that transforms $\rho$ into a single-qubit magic state, and it is invariant under Clifford unitaries.
Since Theorem~\ref{thm:activation} shows that bound magic in the single-copy regime can be activated using two copies, the Triangle Negativity is not additive.

\begin{acknowledgments}

    We thank Hao Dai, Zhenyu Du, David Gross, Otfried Gühne, Markus Heinrich, Fuchuan Wei, Andreas Winter, Zhenyu Cai, Xuanran Zhu and  Chengkai Zhu for insightful discussions. Q.Y. is supported by the National Natural Science Foundation of China (No. T24B2002). Z.-W.L.\ acknowledges support from NSFC under Grant No.~12475023, Dushi Program, and a startup funding from YMSC.
\end{acknowledgments}

\bibliography{bibliography}

\onecolumngrid
\appendix
\clearpage

\section{Triangle Criterion}~\label{app:triangle}
In this section, we aim to prove the validity of Triangle Criterion and calculate the number of sets containing three neighbouring pure stabilizer states, i.e., $\{\psi_i,\psi_j,\psi_k\}$ with $\Tr(\psi_i\psi_j)=\Tr(\psi_i\psi_k)=\Tr(\psi_j\psi_k)=\frac{1}{2}$. 
They are based on the same observation of stabilizer pure states.

\subsection{Proof}

To prove the validity of the Triangle Criterion, we first establish the following two results.
\begin{lemma}\label{lemma:triangle1}
Let $\psi_1$ and $\psi_2$ be two nearest-neighbour stabilizer states, namely $\Tr(\psi_1\psi_2)=\frac{1}{2}$.
Then for any stabilizer state $\phi$, whenever both overlaps are non-zero, the ratio $\frac{\Tr(\psi_1\phi)}{\Tr(\psi_2\phi)}$ can only take the values $\{\frac{1}{2},1,2\}$.
\end{lemma}

\begin{proof}
A standard fact about stabilizer states is that any two nearest-neighbour stabilizer states can be mapped, by a Clifford unitary, to the canonical pair $\ket{0^n}$ and $\ket{+}\otimes\ket{0^{n-1}}$. 
Since Clifford unitaries preserve the set of stabilizer states as well as all fidelities between them, it suffices to prove the claim for
\begin{equation}
|\psi_1\rangle = |0^n\rangle,\qquad
|\psi_2\rangle = |+\rangle\otimes |0^{n-1}\rangle .
\end{equation}

Now let $|\phi\rangle$ be any $n$-qubit stabilizer state. Using the standard normal form of stabilizer states in the computational basis, one can write
\begin{equation}\label{eq:stab_decomp}
|\phi\rangle
=
2^{-m/2}\sum_{x\in A}a_x |x\rangle,
\end{equation}
where $A\subseteq \mathbb{F}_2^n$ is an affine subspace of the $n$-dimensional binary vector space $\mathbb{F}_2^n$, with $|A|=2^m$, and each coefficient satisfies $|a_x|=1$.

We first compute the overlap with $|\psi_1\rangle=|0^n\rangle$. Clearly, when it is non-zero, $\Tr(\psi_1\phi)=2^{-m}$.
Next, since $\ket{\psi_2}=\frac{1}{\sqrt2}\bigl(|0^n\rangle+|10^{n-1}\rangle\bigr)$, we have $\langle \psi_2|\phi\rangle=\frac{1}{\sqrt2}\bigl(\langle 0^n|\phi\rangle+\langle 10^{n-1}|\phi\rangle\bigr)$
There are three possibilities.

\medskip
\noindent
\textbf{Case 1:} $10^{n-1}\notin A$.  
Then, assuming $\langle \psi_1|\phi\rangle\neq 0$, we must have $0^n\in A$, and therefore $\Tr(\psi_2\phi) = 2^{-m-1}.$
Thus $\frac{\Tr(\psi_1\phi)}{\Tr(\psi_2\phi)}=2$.

\medskip
\noindent
\textbf{Case 2:} $10^{n-1}\in A$.  
Then both $0^n$ and $10^{n-1}$ belong to $A$, and
\begin{equation}
\langle \psi_2|\phi\rangle
=
\frac{2^{-m/2}}{\sqrt2}\bigl(a_{0^n}+a_{10^{n-1}}\bigr).
\end{equation}
For a stabilizer state, the relative phase between two basis elements in the support is always a fourth root of unity, so $a_{10^{n-1}}/a_{0^n}\in\{\pm1,\pm i\}$.
Therefore
\begin{equation}
|a_{0^n}+a_{10^{n-1}}|^2
=
\begin{cases}
4, & a_{10^{n-1}}=a_{0^n},\\
2, & a_{10^{n-1}}=\pm i\,a_{0^n},\\
0, & a_{10^{n-1}}=-a_{0^n}.
\end{cases}
\end{equation}
Hence, whenever $\langle \psi_2|\phi\rangle\neq 0$, $\Tr(\psi_2\phi)=|\langle \psi_2|\phi\rangle|^2\in\{2^{-m+1}, 2^{-m}, 2^{-m-1}\}.$
Since $|\langle \psi_1|\phi\rangle|^2=2^{-m}$ whenever it is non-zero, the ratio $\frac{\Tr(\psi_1\phi)}{\Tr(\psi_2\phi)}$ can only be one of the three values in $\{\frac{1}{2}, 1, 2\}$, which concludes the proof.
\end{proof}

\begin{lemma}\label{lemma:triangle2}
Let $\psi_1$, $\psi_2$, and $\psi_3$ be three nearest-neighbour stabilizer states, namely $\Tr(\psi_i\psi_j)=\frac{1}{2}$ for any $i\neq j$ and $\{i,j\}\in\{1,2,3\}$.
Then for any stabilizer state $\phi$, if $\Tr(\psi_i\phi)=0$, then $\Tr(\psi_j\phi)=\Tr(\psi_k\phi)$.
\end{lemma}
\begin{proof}
The proof follows from a similar logic as the last lemma.
With some Clifford unitary, the three neighbouring stabilizer states can be rotated into $\ket{\psi_1}=\ket{0^n}$, $\ket{\psi_2}=\ket{+}\ket{0^{n-1}}$, and $\ket{\psi_3}=\ket{\pm i}\ket{0^{n-1}}$.
Given another stabilizer state $\phi$ defined in \eqref{eq:stab_decomp}, we can assume $\Tr(\psi_1\phi)=0$ without loss of generality.
This means that $0^{n}$ is not in $A$.
We also consider the two cases mentioned in the last proof.

\medskip
\noindent
\textbf{Case 1:} $10^{n-1}\notin A$.  
Then, we have that $\Tr(\psi_1\phi)=\Tr(\psi_2\phi)=0$.

\medskip
\noindent
\textbf{Case 2:} $10^{n-1}\in A$.  
Then we have $\Tr(\psi_1\phi)=\Tr(\psi_2\phi)=2^{-m-1}$.
\end{proof}

With the results of these two lemmas, we can conclude that for any three pairwise neighboring stabilizer pure states $\psi_1,\psi_2,\psi_3$ and any stabilizer pure state $\phi$, the triple of fidelities can only take one of the following forms, up to permutation:
\begin{equation}
\bigl(\Tr(\psi_1\phi),\Tr(\psi_2\phi),\Tr(\psi_3\phi)\bigr)
\in
2^{-m}\left\{
(1,1,1),\,
\left(1,1,\frac12\right),\,
\left(1,\frac12,\frac12\right),\,
(1,1,0),\,
(0,0,0)
\right\}.
\end{equation}
Hence,
\begin{equation}
\Tr(\psi_i\phi)\le \Tr(\psi_j\phi)+\Tr(\psi_k\phi)
\end{equation}
holds for every permutation $(i,j,k)$ of $(1,2,3)$. Since every mixed stabilizer state is a convex combination of stabilizer pure states, and the above inequality is linear in $\phi$, it follows immediately that the same inequality also holds for all mixed stabilizer states.

\subsection{Number of neighbouring stabilizer sets}\label{app:number}
First, we can fix a stabilizer state $\psi_0$ and calculate the number of neighbouring sets containing it.
Since any two stabilizer states can be transformed into each other with some Clifford unitary, this number is the same for all stabilizer states.
Therefore, the total number of different neighbouring sets is equivalent with the number of stabilizer states times the number of neighbouring sets containing a fixed stabilizer state divided by three.
Without loss of generality, we can fix $\ket{\psi_0}=\ket{0^n}$.
If another state $\psi_1$ satisfies that $\bra{0^n}\psi_1\ket{0^n}=\frac{1}{2}$, it must have the form of
\begin{equation}
\ket{\psi_1}=\frac{1}{\sqrt{2}}\left(\ket{0^n}+a\ket{x}\right),
\end{equation}
where $x \in \{0,1\}^n\setminus\{0^n\}$ is an $n$-bit string and $a$ is a phase taking one of the four different values $\{\pm1,\pm i\}$.
It is easy to prove that, given $\psi_0$ and $\psi_1$, the third state $\psi_2$ can only take the form of 
\begin{equation}
\ket{\psi_2}=\frac{1}{\sqrt{2}}\left(\ket{0^n}+a^\prime\ket{x}\right)
\end{equation}
with $aa^\prime\in\{\pm i\}$.
So, in total, the number of different choices of $\psi_1$ and $\psi_2$ is $4(2^n-1)$ and thus the number of different neighbouring sets is 
\begin{equation}
\frac{4}{3}(2^n-1)\times\#\{\mathrm{Stabilizer \ States}\}=\frac{4}{3}(2^n-1)2^n\prod_{i=1}^n(2^k+1)=2^{\mathcal{O}(n^2)}.
\end{equation}

\section{Magic distillation and generation}
\subsection{Equivalence between two formulations}~\label{app:equivalence}
Here we prove the equivalence between Triangle Criterion and the following characterization.

\begin{lemma}\label{lem:witness}
    Given an $n$-qubit state $\rho$, it can be detected by Triangle Criterion if and only if there exists a Clifford unitary $U$ such that $\Tr(W_U\rho)<0$ with
    \begin{equation}
        W_U=U\left[(\ketbra{+}+\ketbra{+i}-\ketbra{0})\otimes\ketbra{0^{n-1}}{0^{n-1}}\right]U^\dagger.
    \end{equation}
\end{lemma}

\begin{proof}
    We first prove the ``if" part. Suppose $\Tr(W_U\rho)<0$ for some Clifford unitary $U$. Let $\psi_1, \psi_2, \psi_3$ be three stabilizer states obtained by applying $U$ to $\ketbra{0}\otimes \ketbra{0^{n-1}}, \ketbra{+}\otimes \ketbra{0^{n-1}}, \ketbra{+i}\otimes \ketbra{0^{n-1}}$, respectively. Then $\Tr(\psi_1\psi_2)=\Tr(\psi_2\psi_3)=\Tr(\psi_3\psi_1)=1/2$ and $\Tr(\rho\psi_1)>\Tr(\rho\psi_2)+\Tr(\rho\psi_3)$, so $\rho$ can be detected by Theorem~\ref{thm:magic_PPT}.

    Now we prove the ``only if'' part. Suppose $\Tr(\rho\psi_1)>\Tr(\rho\psi_2)+\Tr(\rho\psi_3)$ for three pure stabilizer states $\psi_1, \psi_2, \psi_3$ satisfying $\Tr(\psi_1\psi_2)=\Tr(\psi_2\psi_3)=\Tr(\psi_3\psi_1)=1/2$. It is well-known that there is a Clifford unitary $U_1$ such that $U_1\ket{\psi_1}=\ket{0^n}$. Since unitary transformation preserve fidelity, the squared fidelity between $\ket{0^n}$ and $U_1\ket{\psi_2}$ is $1/2$. According to the statement in the last section, $U_1\ket{\psi_2}$ has the form $\frac{1}{\sqrt{2}}(\ket{0^n}+\alpha\ket{x})$ (up to a global phase), where $x\in\{0,1\}^n\setminus\{0^n\}$ is an $n$-bit string and $\alpha$ is a phase taking one of the four different values $\{\pm1,\pm i\}$. 
    Similarly, we can write $U_1\ket{\psi_3}\propto\frac{1}{\sqrt{2}}(\ket{0^n}+\alpha^\prime\ket{x})$ with $\alpha\alpha^\prime=\pm i$. 
    There exists a CNOT circuit $U_2$ that maps $\ket{x}$ to $\ket{10^{n-1}}$, while keeping $\ket{0^n}$ untouched. Therefore, $U_2U_1$ maps $\psi_1, \psi_2, \psi_3$ to $\ket{0^n}, \frac{1}{\sqrt{2}}(\ket{0}+\alpha\ket{1})\ket{0^{n-1}}, \frac{1}{\sqrt{2}}(\ket{0}+\alpha^\prime\ket{1})\ket{0^{n-1}}$ up to global phases, respectively. 
    Let $U_3$ be a single qubit unitary mapping $\frac{1}{\sqrt{2}}(\ket{0}+\alpha\ket{1})$ and $\frac{1}{\sqrt{2}}(\ket{0}+\alpha^\prime\ket{1})$ to $\ket{+}$ and $\ket{+i}$, respectively.
    Then, $U=U_3U_2U_1$ maps $\psi_1$ to $\ketbra{0^n}$, and maps $\{\psi_2, \psi_3\}$ to $\{\ketbra{+}\otimes \ketbra{0^{n-1}}, \ketbra{+i}\otimes \ketbra{0^{n-1}}\}$. Hence, $\Tr(\rho \psi_1)>\Tr(\rho \psi_2)+\Tr(\rho \psi_3)$ implies $\Tr(W_{U^\dagger}\rho)<0$. 
\end{proof}

\subsection{Triangle Criterion for channel property testing}~\label{app:channel}
Here we present the proof for Corollary~\ref{coro:channel}. We first show that if the Choi state can be detected by the Triangle Criterion, then the channel can be used to produce a single-qubit magic state with stabilizer input states and stabilizer operations. The central observation is that, by the definition of the Choi state,
\begin{equation}
    \rho_{\mathcal{C}}=\mathcal{I}_n\otimes\mathcal{C}(\Phi_n^+),
\end{equation}
preparing the Choi state consists of only two steps: (i) preparing the $2n$-qubit maximally entangled state; and (ii) applying the channel $\mathcal{C}$ to the $n$-qubit subsystem. Since the maximally entangled state is itself a stabilizer state, the preparation of the Choi state with query access to the channel does not require additional magic resource. Therefore, if the Choi state can be detected by the $(n+m)$-qubit Triangle Criterion, then one can use $\mathcal{C}$ to produce a single-qubit magic state by first preparing its Choi state and then transforming that Choi state into a single-qubit magic state via Theorem~\ref{thm:bound_magic}.

\begin{figure}
    \centering
    \includegraphics[width=0.25\linewidth]{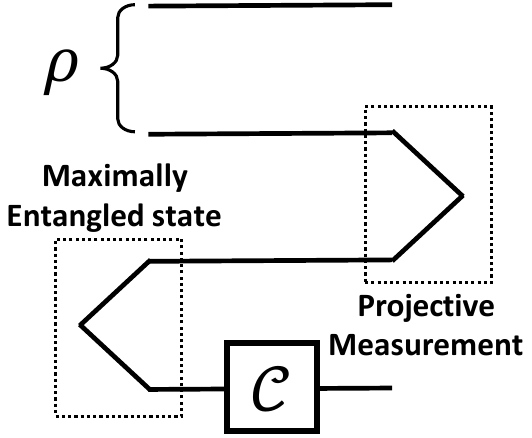}
    \caption{The left dashed black box represents the input maximally entangled state $\Phi^+_n$, and the right dashed black box represents the projective measurement onto $\Phi^+_n$. The first and fourth legs at the right side represents the output state $\mathcal{I}\otimes\mathcal{C}(\rho)$. }
    \label{fig:choi}
\end{figure}

To prove the converse direction, we introduce another observation. The action of $\mathcal{C}$ on any input state can also be represented in terms of its Choi state:
\begin{equation}
\mathcal{I}\otimes\mathcal{C}(\rho)=2^{2n}\Tr_{n,n}[(\rho\otimes\rho_{\mathcal{C}})(I\otimes\Phi^+_n\otimes I)],
\end{equation}
where $I$ denotes the identity operator and $\Tr_{n,n}$ denotes the partial trace over the $n$-qubit subsystem of $\rho$ on which the channel acts and the $n$-qubit subsystem of $\rho_{\mathcal{C}}$ introduced through the maximally entangled state. This equation can be understood graphically and operationally from Fig.~\ref{fig:choi}, where a Bell-state measurement (BSM) is performed between the second subsystem of $\rho$ and the first subsystem of $\rho_{\mathcal{C}}$. Only the outcome corresponding to the projection onto $\Phi^+_n$ is kept, while all other outcomes are discarded. Note that when the input state $\rho$ is a stabilizer state, the preparation of $\mathcal{I}\otimes\mathcal{C}(\rho)$ can be equivalently viewed as a stabilizer operation acting on the Choi state $\rho_{\mathcal{C}}$, since the BSM is itself a stabilizer operation. Therefore, if the channel can produce a single-qubit magic state, then there exists a stabilizer operation that transforms the Choi state $\rho_{\mathcal{C}}$ into a single-qubit magic state. By Theorem~\ref{thm:bound_magic}, this implies that $\rho_{\mathcal{C}}$ can be detected by the Triangle Criterion, which completes the proof.

\subsection{Multi-qubit magic distillation}\label{app:single_distillation}

We now give a two-qubit state $\rho$ with the following properties: 
Using a single copy of $\rho$, one cannot detect magic using Triangle Criterion, and thus one cannot distil the state to a single-qubit magic state via stabilizer operations. 
Yet, using two copies $\rho^{\otimes 2}$, it can be detected with the Triangle Criterion and thus can be reduced to a single-qubit magic state $\rho^\prime$. 
Further, we find that this $\rho^\prime$ can be used for magic state distillation via established distillation protocols~\cite{bravyi2005universal,reichardt2006quantum}. The two-qubit state $\rho$ is given by
\begin{equation}
\begin{aligned}
    \rho=\frac{1}{4}(
    I\otimes I+\frac{1}{\sqrt{12}}(
    &I\otimes X
    + \frac{1}{\sqrt{5}} I\otimes Y  
    + \frac{\sqrt{5}}{2} I\otimes Z
    +\frac{1}{4} X\otimes I
    + X\otimes X
    - \frac{\sqrt{5}}{2} X\otimes Y
    -\frac{1}{2} X\otimes Z
    + \frac{\sqrt{5}}{2} Y\otimes I \\
    &- \frac{1}{\sqrt{5}} Y\otimes X 
     +\frac{1}{2}   Y\otimes Y
    -\frac{1}{4}   Y\otimes Z
     +\frac{\sqrt{5}}{8} Z\otimes I
    +\frac{1}{4}   Z\otimes X
    - \frac{\sqrt{5}}{2} Z\otimes Y
+\frac{\sqrt{5}}{8} Z\otimes Z
    ))
\end{aligned}
\end{equation}
To distil the single-qubit state from $\rho^{\otimes 2}$, we use the circuit shown in Fig.~\ref{fig:distill_circuit}. The success probability of the protocol is $P_{\text{success}}\approx 0.129$.
The distilled single-qubit state can be approximately written as
\begin{equation}
  \rho'\approx \frac{1}{2}(I +  0.1844 X + 0.3334 Y + 0.6544 Z)
\end{equation}
which yields
\begin{equation}
\abs{\Tr(X\rho^\prime)}+\abs{\Tr(Y\rho^\prime)}+\abs{\Tr(Z\rho^\prime)}\approx1.172>\frac{3}{\sqrt{7}}
\end{equation}
meaning it can be distilled into the magic state $\ketbra{H}{H}=\frac{1}{2}(I+\frac{1}{\sqrt{3}}(X+Y+Z))$ with the protocol as introduced in Ref.~\cite{bravyi2005universal} (see also Ref.~\cite{reichardt2006quantum}).
Together with stabilizer operations, $\ketbra{H}{H}$ can also realize universal quantum computation and thus is capable of producing $T$ states.

\begin{figure}[htbp]
	\centering	
    \subfigimg[width=0.18\textwidth]{}{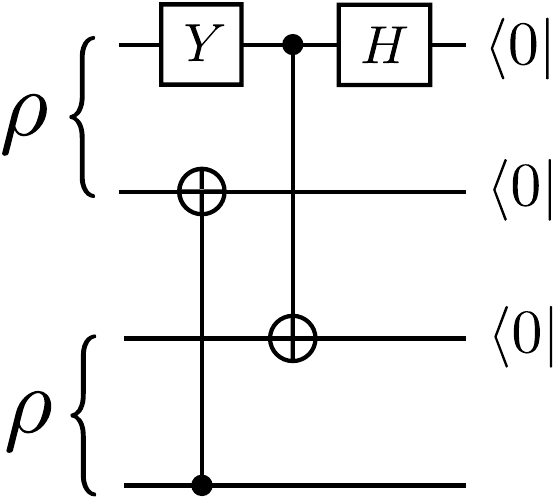}
	\caption{Circuit to distil $\rho^{\otimes 2}$ to a single qubit state $\rho^\prime$, which can be further distilled using magic state distillation protocols. This circuit uses two CNOT gates, a Hadamard gate, and a Pauli-$Y$ gate. }
        \label{fig:distill_circuit}
\end{figure}

\section{The minimal purity of magic state}\label{app:purity}
In this section, we explicitly construct a state $\rho$ with purity $1/(d-1/2)$ that lies at the boundary of the stabilizer polytope:
\begin{align}
    \rho &= \frac{1/2}{d-1/2}(\ketbra{0}+\ketbra{+}+\ketbra{+i})\otimes \ketbra{0^{n-1}} + \frac{1}{d-1/2}\sum_{z\in \{0, 1\}^n \setminus \{0^n, 10^{n-1}\}} \ketbra{z}\\
    &= \frac{1}{d-1/2}I^{\otimes n} + \frac{1/2}{d-1/2}(\ketbra{0}+\ketbra{+}+\ketbra{+i}-2I)\otimes \ketbra{0^{n-1}}. \label{eq:rho_upper_bound}
\end{align}
Direct calculation gives $\Tr(\rho^2)=1/(d-1/2)$. To prove that $\rho$ lies at the boundary of the stabilizer polytope, we only need to show that for any stabilizer state $\sigma$, 
\begin{equation}
    \Tr(\rho\sigma)\leq \Tr(\rho^2)=\frac{1}{d-1/2}. \label{eq:upper_bound_constraint}
\end{equation}
And the equality can be satisfied by some $\sigma$.
By \eqref{eq:rho_upper_bound}, 
\begin{equation}
    \Tr(\rho\sigma) = \frac{1}{d-1/2}+\frac{1/2}{d-1/2}\Tr\left[(\ketbra{0}+\ketbra{+}+\ketbra{+i}-2I)\sigma_0\right],
\end{equation}
where $\sigma_0=\left(I\otimes \bra{0^{n-1}}\right)\sigma\left(I\otimes \ket{0^{n-1}}\right)$. It is known that $\sigma_0$ is proportional to a single-qubit stabilizer state. Enumerating $\sigma_0\propto \ketbra{0}, \ketbra{1}, \ketbra{+}, \ketbra{-}, \ketbra{+i}, \ketbra{-i}$, we find that $\Tr(\rho\sigma)\leq 1/(d-1/2)$ in all cases and the equality can be satisfied when $\sigma_0\propto \ketbra{0}{0}, \ketbra{+}{+}, \ketbra{+i}{+i} $. 

This argument can be generalized to prove the following theorem:
\begin{theorem}\label{thm: a_n non-increasing}
    Denote the infimum of all possible purity values of $n$-qubit magic states by $r_n$. Then $\{2^n-1/r_n\}$ is non-increasing. Equivalently, if we write $r_n=1/(2^n-a_n)$, then $\{a_n\}$ is non-increasing.
\end{theorem}
\begin{proof}
    Fix $n<m$. Write $m=n+t$. Let $\rho_n$ be an $n$-qubit state on the boundary of the stabilizer polytope with purity $r_n$. Consider the following state:
    \begin{equation}
        \rho_m = \frac{1}{(2^{m}-2^n)r_n+1} \rho_n\otimes \ketbra{0^t} + \frac{r_n}{(2^{m}-2^n)r_n+1}\left(I_{m}-I_{n}\otimes \ketbra{0^t}
        \right).
    \end{equation}
    It is easy to verify that $\rho_m\ge 0$ and
    \begin{equation}
        \Tr(\rho_m) = \frac{1}{(2^{m}-2^n)r_n+1}+\frac{r_n}{(2^{m}-2^n)r_n+1}(2^m-2^n)=1.
    \end{equation}
    Therefore, $\rho_m$ is a valid quantum state. Furthermore, since $\rho_n$ is a mix of $n$-qubit stabilizer states, $\rho_m$ is a mix of $m$-qubit stabilizer states, thus $\rho_m$ lies in the $m$-qubit stabilizer polytope. We now prove that $\rho_m$ indeed lies on the boundary of the stabilizer polytope by showing that $\Tr(\rho_m\sigma)\leq \Tr(\rho_m^2)$ for all $m$-qubit stabilizer states $\sigma$. By definition,
    \begin{equation}
        \Tr(\rho_m^2) = \frac{r_n}{[(2^{m}-2^n)r_n+1]^2} + \frac{r_n^2}{[(2^{m}-2^n)r_n+1]^2}(2^m-2^n)=\frac{r_n}{(2^m-2^n)r_n+1},
    \end{equation}
    \begin{equation}
        \Tr(\rho_m\sigma) = \frac{r_n}{(2^m-2^n)r_n+1}+\frac{1}{(2^m-2^n)r_n+1}(\Tr(\rho_m\sigma_0)-r_n\Tr(\sigma_0)),
    \end{equation}
    \begin{equation}
        \Tr(\rho_m^2)-\Tr(\rho_m\sigma) = \frac{1}{(2^m-2^n)r_n+1}(r_n\Tr(\sigma_0)-\Tr(\rho_m\sigma_0)), \label{eq: a_n non-increasing proof}
    \end{equation}
    where $\sigma_0=(I_n\otimes \bra{0^t})\sigma(I_n\otimes \ket{0^t})$. 
    As the projection towards $\ket{0^t}$ is a stabilizer operation, $\sigma_0=r\sigma'$ for some $r\ge 0$ and $n$-qubit stabilizer state $\sigma'$. Since $\rho_n$ is a state on the boundary of $n$-qubit stabilizer polytope, the facet containing $\rho_n$ is perpendicular to $\rho_n$, thus $\Tr(\rho_n\sigma')\leq \Tr(\rho_n^2)=r_n$. According to~\eqref{eq: a_n non-increasing proof}, $\Tr(\rho_m^2)-\Tr(\rho_m\sigma)\ge 0$, so $\rho_m$ is on the boundary of the stabilizer polytope. By definition of $r_m$, $r_m\leq \Tr(\rho_m^2)$, implying that
    \begin{equation*}
        r_m\leq \frac{r_n}{(2^m-2^n)r_n+1}\Rightarrow 2^m-\frac{1}{r_m}\leq 2^n-\frac{1}{r_n}.\qedhere
    \end{equation*}
\end{proof}

In Table 2 and Eq.~(22) of Ref.~\cite{reichardt2006quantum}, the author listed all eight different kinds of surfaces in two-qubit stabilizer polytope.
After calculation, we find that the minimal purities on these eight different surfaces are $\frac{3}{10}$, $\frac{11}{36}$, $\frac{5}{16}$, $\frac{11}{36}$, $\frac{7}{20}$, $\frac{4}{13}$, $\frac{43}{140}$, and $\frac{2}{7}$ respectively.
The minimal purity value is $\frac{2}{7}$, satisfying the conjecture of $\frac{1}{d-1/2}$ with $d=4$.

Unfortunately, we currently only managed to formally establish the following lower bound, which illustrates two different approaches that are insufficient to yield our conjecture.
\begin{theorem}
If the purity of a $d$-dimensional state is lower than $\frac{1}{d-1/d}$, it is a mixture of stabilizer states.
\end{theorem}
\begin{proof}
We here provide two different proofs.

The first proof of this theorem is based on a simple observation, that $I\pm P$ with $P$ being an arbitrary Pauli operator is proportional to a stabilizer state.
Given any density matrix $\rho$, we can decompose it in the Pauli basis as 
\begin{equation}
\rho=\frac{I}{d}+\frac{1}{d}\sum_{i=1}^{d^2-1}t_i P_i
\end{equation}
with the purity being $\Tr(\rho^2)=\frac{1}{d}+\frac{1}{d}\sum_it_i^2$.
If 
\begin{equation}
\sum_{i=1}^{d^2-1}\abs{t_i}\le 1,
\end{equation}
we can rewrite the density matrix as
\begin{equation}
\rho=\frac{1}{d}\left(1-\sum_i \abs{t_i}\right)I+\frac{1}{d}\sum_i\abs{t_i}\left(I+\mathrm{sign}(t_i)P_i\right).
\end{equation}
As every term at the R.H.S. is stabilizer state, the mixture of them is also a stabilizer state.
It is easy to prove that when $\abs{t_i}=\frac{1}{d^2-1}$ for all $i$, the purity reaches its minimum, being
\begin{equation}
\Tr(\rho^2)=\frac{1}{d}+(d^2-1)\times\frac{1}{(d^2-1)^2}\times \frac{1}{d}=\frac{1}{d}+\frac{1}{d(d^2-1)}=\frac{1}{d-1/d}.
\end{equation}

The second proof uses the conclusion that we can construct $d+1$ different mutually unbiased bases with $d(d+1)$ stabilizer state~\cite{Zhu2015MUBClifford2Design}.
Therefore, the minimal purity of states outside of the polytope consisted with these $d(d+1)$ states gives a lower bound for the magic states.
In Ref.~\cite{Bengtsson2005ComplementarityPolytope}, the authors calculated the inner radius of the polytope, with the corresponding purity also being $\frac{1}{d-1/d}$.

\end{proof}

\section{Absolute non-resourceful ball}\label{app:absolute}

It is known that the purity of a given state does not change under unitary evolution.
In this section, we will prove that the purity function also has certain stability even under measurement and post-selection, i.e., Theorem~\ref{thm:purity_reduction}.
\begin{proof}
First, to calculate the maximal purity, it is sufficient to only consider rank-1 projective measurement in system $B$, as other measurements result in a mixture of rank-1 measurement state in system $A$.
Therefore, we can write the measurement operator as 
\begin{equation}
\Pi=I_A\otimes\ketbra{\psi_B}{\psi_B}=\begin{bmatrix}
    I_S & 0 \\ 0 & 0
\end{bmatrix},
\end{equation}
which is acted on the rotated state written in the same basis
\begin{equation}
U_{AB}\rho_{AB}U_{AB}^\dagger=\begin{bmatrix}
    \rho_S & X \\ X^\dagger & \rho_R,
\end{bmatrix},
\end{equation}
where $I_S, \rho_S$ are $d_A\times d_A$, $\rho_R$ is $(d-d_A)\times (d-d_A)$. Our target is to upper bound $\Tr(\rho_S^2)/\Tr(\rho_S)^2$.
Define
\begin{equation}
q = \Tr(\rho_S), \ P_R=\Tr(\rho_R^2), \ P_S=\Tr(\rho_S^2).
\end{equation}
It is easy to prove that 
\begin{equation}
\frac{q^2}{d_A}\le P_S \le q^2, \frac{(1-q)^2}{d-d_A}\le P_R\le (1-q)^2, \ P_S+P_R=\frac{1}{d-c}-\Tr(XX^\dagger)-\Tr(X^\dagger X)\leq \frac{1}{d-c}.
\end{equation}
So we have
\begin{equation}
    \frac{\Tr(\rho_S^2)}{\Tr(\rho_S)^2}\leq \frac{1}{q^2}\left(\frac{1}{d-c}-P_R\right)\leq \frac{1}{q^2}\left(\frac{1}{d-c}-\frac{(1-q)^2}{d-d_A}\right)\coloneqq f(q).
\end{equation}
Recognizing the $f(q)$ as a quadratic function of $1/q$, we can see it reaches the maximum when $q=(d_A-c)/(d-c)$. Hence,
\begin{equation}
    \frac{\Tr(\rho_S^2)}{\Tr(\rho_S)^2}\leq f(q)\leq f\left(\frac{d_A-c}{d-c}\right) = \frac{1}{d_A-c}.
\end{equation}
\end{proof}

\section{Detectability of magic states}
\subsection{Unfaithful magic state}\label{app:faithful}
Consider the magic witness of the form $W=\alpha I-\ketbra{\psi}{\psi}$ with $\psi$ being a magic pure state.
To make sure that $W$ is a magic witness, we require that for all stabilizer state $\rho$, it satisfies
\begin{equation}
\Tr(W\rho)=\alpha-\bra{\psi}\rho\ket{\psi}\ge 0.
\end{equation}
Therefore, we can set $\alpha\coloneqq\max_{\rho\in \mathrm{Stab}}\bra{\psi}\rho\ket{\psi}$.
Next, we want to calculate the minimal value of $\alpha$ for all magic state $\psi$.

\begin{lemma}
The value of $\alpha=\min_{\psi}\left[\max_{\rho\in\mathrm{Stab}}\bra{\psi}\rho\ket{\psi}\right]$ is lower bounded by $\frac{3}{d+2}$, where $d$ is the dimension of $\psi$.
\end{lemma}
\begin{proof}
First, as $\bra{\psi}\rho\ket{\psi}$ is a linear function in $\rho$, it is sufficient to only consider pure stabilizer states $\{\phi_i\}_i$.
Define $x_i\coloneqq\abs{\bra{\psi}\ket{\phi_i}}^2$ and $X=\max_i x_i$.
Use the fact that the stabilizer states form a state 3-design, we have 
\begin{equation}
\mathbb{E}_ix_i=\frac{1}{d}, \ \ \ \mathbb{E}_i x_i^2=\frac{2}{d(d+1)}, \ \ \ \mathbb{E}_i x_i^3=\frac{6}{d(d+1)(d+2)}.
\end{equation}
By definition, $x_i^3\le Xx_i^2$ and thus $\mathbb{E}_ix_i^3\le X\mathbb{E}_ix_i^2$.
We have
\begin{equation}
X\ge \frac{\mathbb{E}_ix_i^3}{\mathbb{E}_ix_i^2}=\frac{3}{d+2},
\end{equation}
which concludes the proof.
\end{proof}

This conclusion means that, given a magic state $\rho$, if it can be detected by some fidelity-based magic witness, the fidelity between it and some magic pure state should be larger than $\frac{3}{d+2}$.
We will show that this is not possible for the magic state we find whose purity is arbitrarily close to $\frac{1}{d-1/2}$.
Any state $\rho$ can be decomposed as $\rho=\frac{I}{d}+t\sigma$ with $t$ being a positive coefficient and $\sigma$ being a Hermitian matrix satisfying $\Tr(\sigma)=0$ and $\Tr(\sigma^2)=1$.
Due to the purity condition, we have 
\begin{equation}
t=\sqrt{\frac{1}{d-1/2}-\frac{1}{d}}=\sqrt{\frac{1/2}{d(d-1/2)}}.
\end{equation}
Due to the requirements for $\sigma$, we have $\bra{\psi}\sigma\ket{\psi}\le\frac{1}{\sqrt{2}}$.
We thus have
\begin{equation}
\bra{\psi}\rho\ket{\psi}\le\frac{1}{d}+t\bra{\psi}\sigma\ket{\psi}\le\frac{1}{d}+\frac{1}{2}\sqrt{\frac{1}{d(d-1/2)}},
\end{equation}
which is less than $\frac{3}{d+2}$ when $d>2$.

\subsection{Detectability of Triangle Criterion and multi-qubit magic}~\label{app:detectability}
In Ref.~\cite{liu2022detectability}, the authors proved a conclusion regarding the detectability of a witness operator (although the original statement is about entanglement witness, it actually works for all witness operators):
\begin{fact}[Theorem 2 of Ref.~\cite{liu2022detectability}]\label{fact:witness}
Given states sampled according to distribution $\pi_{d,k}$ and witness operator satisfying $\Tr(W)>0$, the probability for successfully detection is upper bounded by
\begin{equation}
\underset{\rho\sim\pi_{d,k}}{\mathrm{Pr}}\left[\Tr(W\rho)<0\right]<2\exp{-\left(\sqrt{1+\frac{\Tr(W)}{\sqrt{\Tr(W^2)}}}-1\right)k}.
\end{equation}
\end{fact}
We test this numerically using the Triangle witness $W_{ijk}=\psi_i+\psi_j-\psi_k$ and show the results in Fig.~\ref{fig:prob}. 
For different qubit number $n$, the probability of a single witness operator detecting magic decays exponentially with $k$. 
Notably, the probability is independent of $n$, which satisfies the prediction of Fact~\ref{fact:witness} as $\frac{\Tr(W_{ijk})}{\sqrt{\Tr(W_{ijk}^2)}}$ is a constant.
\begin{figure}[htbp]
	\centering	
    \subfigimg[width=0.3\textwidth]{}{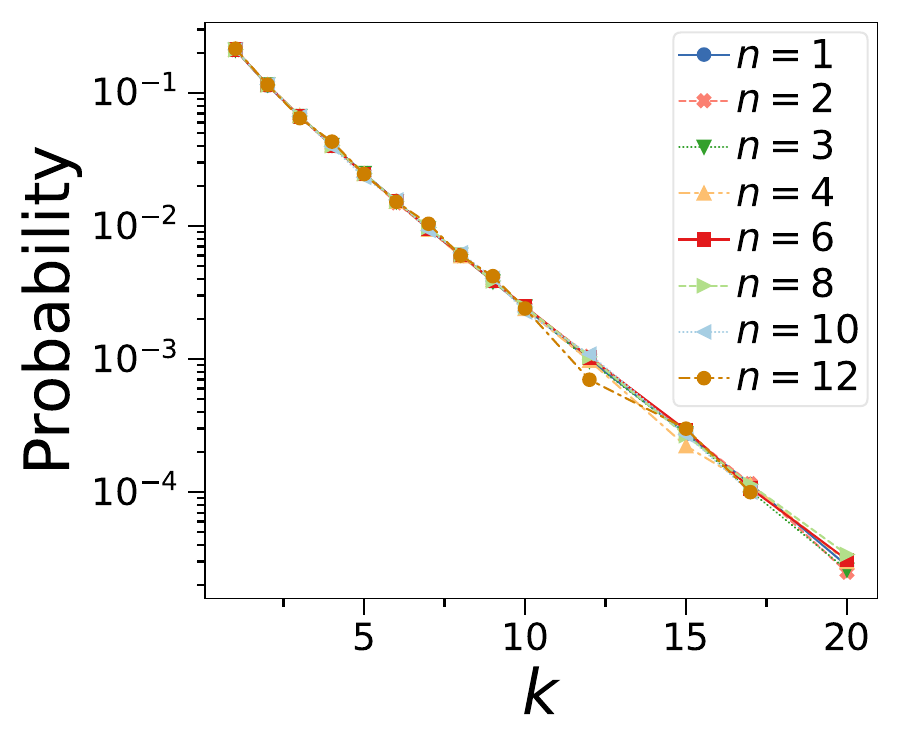}\hfill
	\caption{Probability of detecting magic using a single triangle witness operator $W_{ijk}$, i.e. observing $\Tr(\rho W_{ijk})<0)$. Here, we sample states $\rho\sim \pi_{d,k}$ and plot against traced-out dimension $k$  for different qubit number $n$, where state dimension $d=2^n$. }
        \label{fig:prob}
\end{figure}

Ref.~\cite{liu2022detectability} also generalized their result from a single witness to all linear resource detection methods based on many different linear observables.
Here, by linear detection methods, we mean that the expectation values of a set of $M$ different observables are given, i.e., $\{\Tr(O_i\rho)=r_i\}_{i=1}^M$, and used to decide whether the state have resource or not.
Note that all states with expectation values $\{\Tr(O_i\rho)=r_i\}_{i=1}^M$ also constitute a convex set.
According to the hyperplane separation theorem, this convex set can be separated with the set of the free states with a single witness operator.
As the detection capability of one witness operator has some limitation, one can prove with tools developed in Ref.~\cite{liu2022detectability} that:
\begin{fact}[Theorem 4 of Ref.~\cite{liu2022detectability}]\label{fact:linear_methods}
Given states sampled according to distribution $\pi_{d,k}$ and any linear detection method with $M$ different observables, the probability for successfully detection is upper bounded by
\begin{equation}
2\exp{M\ln{4\sqrt{M}d}-\left(\sqrt{\frac{1}{2}+\min_{W}\frac{\Tr(W)}{\sqrt{\Tr(W^2)}}}-1\right)^2k},
\end{equation}
where $W$ is minimized over all valid resource detection witnesses.
\end{fact}

As proved in \SM{}~\ref{app:number}, the magic Triangle Criterion is equivalent with testing $2^{\mathcal{O}(n^2)}$ different magic witnesses with the form of $W_{ijk}=\psi_i+\psi_j-\psi_k$.
It is easy to verify that $\frac{\Tr(W)}{\sqrt{\Tr(W^2)}}=\frac{1}{\sqrt{2}}$ is a constant.
According to Fact~\ref{fact:witness}, the probability of a single Triangle witness decays exponentially with $k$.
With union bound, we can prove that
\begin{theorem}
Given a $d=2^n$-dimensional state $\rho$ sampled according to the distribution of $\pi_{d,k}$, the probability that it can be detected by the magic Triangle Criterion is upper bounded by $e^{\mathcal{O}(n^2-k)}$.
\end{theorem}
\noindent As the Triangle Criterion is the necessary condition for a state to be distilled by single-qubit magic distillation protocols, $e^{\mathcal{O}(n^2-k)}$ is also the upper bound for single-qubit magic distillation protocols being useful, thus proving Theorem~\ref{thm:single_utility}.

In the main text, we conjecture that the minimal purity of magic state is arbitrarily close to $\frac{1}{d-1/2}$.
We will show that this conjecture gives the lower bound of $\frac{\Tr(W)}{\sqrt{\Tr(W^2)}}$ for magic witness $W$.
\begin{lemma}
If the minimal purity of magic states is arbitrarily close to $\frac{1}{d-1/2}$, the minimal value of $\frac{\Tr(W)}{\sqrt{\Tr(W^2)}}$ is lower bounded by $\frac{1}{\sqrt{2}}$, where $W$ is a valid magic witness satisfying $\Tr(W\rho)\ge0$ for all stabilizer state $\rho$.
\end{lemma}
\begin{proof}
Without loss of generality, we assume that the magic witness satisfies $\Tr(W)=1$.
Thus, the witness can be decomposed in the form of
\begin{equation}
    W = \frac{I}{d}+t\sigma
\end{equation}
with $\Tr(\sigma)=0$ and $\Tr(\sigma^2)=1$.
Then, one can verify that the following is a valid density matrix with purity $\frac{1}{d-1/2}$:
\begin{equation}
\rho_0=\frac{I}{d}-\sqrt{\frac{1/2}{d(d-1/2)}}\sigma.
\end{equation}
Due to the conjecture, this state is also a stabilizer state, so
\begin{equation}
\Tr(W\rho_0)=\frac{1}{d}-t\sqrt{\frac{1/2}{d(d-1/2)}}\ge 0.
\end{equation}
Therefore, we have 
\begin{equation}
\Tr(W^2)=\frac{1}{d}+t^2\le \frac{1}{d}+\frac{2d-1}{d}=2
\end{equation}
and 
\begin{equation}
\frac{\Tr(W)}{\sqrt{\Tr(W^2)}}\ge \frac{1}{\sqrt{2}}.
\end{equation}
\end{proof}

Combined with Fact~\ref{fact:witness} and Fact~\ref{fact:linear_methods}, we can arrive at the conclusion made in the main text:
\begin{theorem}
Assume the minimal purity of magic state is arbitrarily close to $\frac{1}{d-1/2}$.
Given states sampled according to distribution of $\pi_{d,k}$ and any magic witness, the probability for successfully detecting magic is upper bounded by
\begin{equation}
2\exp{-\left(\sqrt{1+\frac{1}{\sqrt{2}}}-1\right)^2k}.
\end{equation}
For any linear detection method with $M$ different observables, the probability for successful detection is upper bounded by 
\begin{equation}
2\exp{M\ln{4\sqrt{M}d}-\left(\sqrt{\frac{1}{2}+\frac{1}{\sqrt{2}}}-1\right)^2k}.
\end{equation}
\end{theorem}

\end{document}